%% file: main.tex
\documentclass[12pt,a4paper]{article}

\usepackage[dvipsnames]{xcolor}

\usepackage{flonat-wp}
\usepackage[harvard]{flonat-bib}

\newcommand{\classes}{\mathcal{C}}
\newcommand{\variants}{\mathcal{V}}
\newcommand{\Strategies}{\mathcal{X}}
\newcommand{\Env}{\mathsf{Env}}
\newcommand{\Frag}{\mathsf{Frag}}
\newcommand{\Mean}{\mathsf{Mean}}
\newcommand{\Harm}{\mathsf{H}}
\newcommand{\Metric}{\mathsf{M}}
\newcommand{\Utility}{\mathsf{U}}
\newcommand{\TrueHarm}{\mathsf{H}^{\star}}
\newcommand{\AudHarm}{\widehat{\mathsf{H}}}
\newcommand{\Dmass}{\Delta}
\newcommand{\classof}[1]{\mathsf{cl}(#1)}

\providecommand{\Description}[1]{}

\begin{document}

\begin{titlepage}

    \title{Gaming the Metric, Not the Harm:\\Certifying Safety Audits against Strategic Platform Manipulation%
        \thanks{This work was supported by the Engineering and Physical Sciences Research Council (EPSRC) Centre for Doctoral Training in Cyber Secure Everywhere [grant number EP/Y035313/1].}}

    \author{Florian A.\ D.\ Burnat\thanks{University of Bath, \texttt{fadb20@bath.ac.uk}.}
        \and Brittany I.\ Davidson\thanks{University of Bath, \texttt{bid23@bath.ac.uk}.}}

    \date{\today}

    \maketitle

    \begin{abstract}
        \noindent Online-safety regulation under the UK Online Safety Act and the EU Digital Services Act increasingly treats scalar metrics as compliance evidence. Once announced, such a metric also becomes an optimization target: a strategic platform can improve its score by routing recommendations through semantically equivalent content variants, without reducing true harm. We ask when such an audit metric can still certify a genuine reduction in harm. The protocol is modeled as a published transformation graph whose connected components form semantic classes, and the metric itself is treated as a security object. Three results follow. First, any metric that scores variants directly is manipulable as soon as two equivalent variants in a harmful class disagree in score. Second, the \emph{semantic-envelope lift}, which assigns each variant the maximum score in its class, is the unique pointwise minimum among conservative classwise-constant repairs. Third, a class-stratified certificate, $H^\star(x) \le (1/\hat\alpha) M_{\mathrm{Env}(m)}(x) + \bar\eta$, holds for every platform strategy, with $\bar\eta$ absorbing annotation and protocol error. We check the claims at three levels: exhaustive enumeration on a finite-state grid of mixed strategies, an SMT encoding in Z3 cross-replayed in cvc5, and a bounded single-player MDP encoded in PRISM-games. The fragile metric fails manipulation invariance and cannot support the same useful predeclared class-coverage certificate; under the envelope-level certificate, it produces large violations at every tested instance, with a large mean gaming gap across random catalogs at a fixed audit budget. The semantic-envelope metric exhibits no such violation in the tested instances.

    \end{abstract}

    \setcounter{page}{0}
    \thispagestyle{empty}

\end{titlepage}

\clearpage


\section{Introduction}

Two recent regulatory documents have made recommender systems an audit target rather than a design choice. Ofcom's child-safety rules name personalized recommendations as a main pathway through which children encounter harmful content online \citep{Ofcom2025-rs}. Ofcom's additional-safety-measures consultation goes further: it recommends collecting safety metrics during on-platform testing of proposed or actual recommender systems (\S14.16) and explicitly identifies recommender-system ``gaming'' as a major risk (\S14.8) \citep{Ofcom2025-asm}.\footnote{Ofcom's framing of ``gaming'' is user-side---bad actors manipulating ranking signals (\S14.8)---whereas this paper extends the same vocabulary to platform-side gaming of the audit metric.} In parallel, the European Commission's 2025 protection-of-minors guidelines under the Digital Services Act position online-safety evaluation as a practical compliance problem for platforms that serve minors \citep{European-Commission2025-qs}.

This creates a security problem, not just a policy problem. Once a platform is judged by an exposure metric, the metric becomes a target. A platform can improve the score by changing what is measured, how it is labeled, or which semantically equivalent content variants are exposed, even when the underlying harmful exposure remains constant. Concrete representation choices a recommender already controls include: with which thumbnail or caption a video is presented, with which paraphrase of a harmful post is surfaced when duplicates exist, with which localized rendering of the same content is shown, with which auto-generated summary or excerpt accompanies the link, the ranking among semantically equivalent posts, and the moderation-label or metadata-field values that the auditing pipeline reads. The resulting failure mode is a form of reward hacking or Goodhart pressure \citep{Skalse2022-oz}. Functional test suites, such as HateCheck, demonstrate the plausibility of this in practice: moderation models are often brittle to keyword choice, spelling variation, and contrastive non-hate cases that share surface cues with hateful content \citep{Rottger2021-hc}. Security has already encountered closely related failures of measurement and evaluation in other domains \citep{Aerni2024-sm,Wang2025-gt}.

Existing work on recommender auditing, causal auditing metrics, and harm-aware recommendations provides useful ingredients \citep{Messmer2023-ar,Sharma2024-yp,Chee2024-nl}. However, much of this literature treats the metric as descriptive: the platform is measured, the score is reported, and the metric is not modeled as an attack surface. This study asks a narrower question:

\begin{quote}
    \emph{When does a recommender-system safety metric certify a non-trivial upper bound on harmful exposure, even if the audited platform strategically manipulates the measurement channel inside a published audit protocol?}
\end{quote}

The scope is intentionally narrow. We isolate a single layer of audit attack surface: \emph{within-class representation choice under a frozen, published protocol}. Specifically, in scope are (i) the platform's choice of distribution over admissible variants, (ii) the metric's invariance and certification properties quantified over all such strategies, and (iii) auditor-side protocol sensitivity to the validation threshold $\rho$. Out of scope, explicitly, are (a) protocol renegotiation or regulator-platform bargaining over $T_0$, $A$, $s$, or $\rho$; (b) adversarial generation of new variants beyond the published transformation family; (c) adversarial labeling or relabeling of the ground truth; (d) classifier poisoning, evasion attacks, or user-side manipulation; (e) sampling error and statistical inference from a field deployment; and (f) multi-platform ecosystems. Each of these is a distinct threat model and requires its own analysis; a complete audit framework must address all of them. This paper conditions on a published protocol that defines which representation changes count as admissible and asks what properties the metric should satisfy \emph{given that protocol}, keeping the contribution squarely at the level of measurement hardening.

We focus on two properties: \emph{manipulation invariance}, that a metric should not improve solely because the platform swaps one content variant for another semantically equivalent variant with a lower classifier score; and \emph{certification}, that the metric should upper-bound harmful exposure via a quantitative inequality $\Harm(x)\leq \gamma \Metric(x)+\beta$ for every admissible platform strategy $x$.

Our novelty claim is correspondingly modest. We do not claim a new general theory of manipulation beyond strategic classification or reward-hacking work. Our contribution is to transport the adversarial lens to \emph{published platform-audit metrics}, to isolate the within-class manipulation surface created by audit protocols, and to show that a simple max-over-class repair is the least conservative fix for that threat model.

The repair is the \emph{semantic envelope}. Once the auditor has published the semantic classes, each variant is scored by the maximum score attained by any admissible variant in its class. This construction has three useful properties: it restores invariance to within-class manipulation, it is the unique pointwise minimum among conservative classwise-constant repairs, and it yields direct classwise certificates whenever harmful classes are bounded away from zero. We then extend the certificate to imperfect class constructions and annotations by introducing an explicit disagreement-mass term, so that uncertainty appears as quantified slack rather than as an invisible assumption.

To make the framework concrete, we do two things. First, we provide a worked public-data protocol example using HateCheck to show what a protocol-defined semantic class looks like in practice. Second, we instantiate the framework using fully reproducible synthetic stress tests. The platform chooses a distribution of recommended items to maximize utility under an audit budget. Harmful classes include ``original'' and ``manipulated'' variants with identical true harm but different audit scores. The experiments are illustrative rather than calibrated to a deployed recommender: utility is used only to choose one adversarial best response, whereas the theorems themselves quantify over all strategies. This separation matters throughout the paper: the proofs are utility-agnostic, whereas the LPs realize one utility-specific selection from a feasible set.

\paragraph{Contributions.}
\begin{itemize}[leftmargin=1.2em]
    \item We formalize \emph{metric robustness under platform-side manipulation} for recommender-system audits, making semantic classes explicit as a published protocol rather than a hidden assumption.
    \item We define manipulation invariance and certification as security properties of audit metrics, prove that direct variant scoring is manipulable, and show that the semantic-envelope lift is the unique pointwise minimum among conservative classwise-constant repairs.
    \item We add a protocol-sensitivity result: when the auditor tightens validation or removes disputed equivalence edges, the induced partition refines and the envelope score weakly decreases pointwise, yielding a monotone sensitivity band over protocol choices.
    \item We extend the certification theorem to imperfect annotation and protocol construction via a disagreement-mass term $\Dmass(x)$, and we show that the resulting certificate is utility-agnostic because it holds for all platform strategies.
    \item We provide a reproducible synthetic evaluation, compare the semantic envelope to a non-conservative class-mean repair, and include a worked HateCheck-based protocol example together with finite-state, SMT, and bounded-MDP consistency checks (a single-player MDP encoded in PRISM-games).
\end{itemize}

The framework is intentionally modest. We do not claim that semantic envelopes solve the full platform-audit problem, nor that our synthetic stress tests establish deployability for Ofcom- or DSA-scale audits. The claim is smaller and, in our view, useful: \emph{audit metrics should be treated as security objects}, and a simple max-over-class repair already removes one well-defined attack surface.

\section{Threat Model and Setup}

We consider an auditor who evaluates a recommender system for a protected population, such as minors. The recommender exposes content variants from a finite set, $\variants$. We assume $\variants$ is finite throughout; this ensures that the semantic classes $\classes$ defined below form a finite partition of $\variants$ and maintains the audit LPs and SMT encodings tractable.

\begin{table}[t]
    \centering
    \caption{Notation reference. Hatted symbols denote audit-side estimates; starred symbols denote latent ground truth; bare symbols apply under the ideal-case classwise harm labeling.}
    \label{tab:notation}
    \footnotesize
    \begin{tabular}{@{}ll@{}}
        \toprule
        Symbol & Meaning \\
        \midrule
        $\variants$ ($V$), $\classes$ ($\mathcal{C}$) & finite variant set; partition into semantic classes \\
        $\classof{v}$ & semantic class containing variant $v$ \\
        $h(c)\in\{0,1\}$ & ideal-case classwise harm label \\
        $h^{\star}(v)\in\{0,1\}$ & latent variant-level harm \\
        $\hat h(c)\in\{0,1\}$ & audit's classwise harm estimate \\
        $m(v)\in[0,1]$ & per-variant detector score \\
        $\Env(m)(v)$ & semantic-envelope lift, $\max_{u:\classof{u}=\classof{v}} m(u)$ \\
        $x\in\Strategies$ ($\mathcal{X}$) & platform strategy, with $\Strategies\subseteq\Delta(\variants)$ \\
        $\Harm(x)$ & ideal-case harmful exposure (under $h$) \\
        $\TrueHarm(x)$ & latent harmful exposure (under $h^{\star}$) \\
        $\AudHarm(x)$ & audit-estimated harmful exposure (under $\hat h$) \\
        $\Metric_m(x)$ & audit metric, $\sum_v x_v m(v)$ \\
        $\Dmass(x)$ & disagreement mass, $\sum_v x_v |h^{\star}(v)-\hat h(\classof{v})|$ \\
        $\hat\alpha$ & coverage, $\min_{c:\hat h(c)=1}\max_{u:\classof{u}=c} m(u)$ \\
        $\bar\eta$ & published worst-case bound on $\Dmass(x)$ \\
        $\tau$ & audit budget on the announced metric \\
        \bottomrule
    \end{tabular}
\end{table}

\begin{definition}[Published protocol and semantic classes]
A published audit protocol specifies four objects:
\begin{enumerate}[label=(\roman*),leftmargin=1.2em]
    \item a finite candidate transformation family $T_0\subseteq \variants\times\variants$ describing which representation changes the audit is willing to consider;
    \item an attribute-preservation checklist $A(v,u)\in\{0,1\}$ stating whether $v$ and $u$ agree on the audit-relevant fields that must be preserved (for example target identity, direction, or speech-act type);
    \item a validation rule or study that assigns each candidate pair a confidence score $s(v,u)\in[0,1]$ that the pair preserves the audited semantics; and
    \item an acceptance threshold $\rho\in[0,1]$.
\end{enumerate}
The published admissible edge set is
\[
    E_{\rho}=\{(v,u)\in T_0 : A(v,u)=1 \ \wedge\ s(v,u)\ge \rho\}.
\]
The semantic classes $\classes$ are the equivalence classes of the reflexive-symmetric-transitive closure of $E_{\rho}$ (the symmetric step ensures that admissibility does not depend on the auditor's choice of which variant in a candidate pair is the ``original''). Transitive closure is mathematically convenient but semantically risky: pairwise ``same semantics'' judgments are noisy, so a chain of individually plausible edges can connect endpoints whose audited harm no longer matches. We treat any class formed via long transitive chains as a candidate for the disagreement-slack regime of \Cref{thm:disagreement} rather than for the ideal-case certificate, and recommend that auditors publish the maximum path length and the weakest-edge confidence per class so that the cost of transitive merging is visible.
\end{definition}

Write $\classof{v}\in\classes$ for the semantic class containing $v$. Human disagreement is not hidden in a black-box predicate; it appears explicitly through $s(v,u)$ and the threshold $\rho$. Varying $\rho$ yields a family of protocol-induced partitions, which we will later use to make the sensitivity to the protocol visible rather than implicit.

For the core theorems, we first analyze the ideal case in which the published classes are harm-pure and correctly labeled.

\begin{definition}[Ideal class purity]
    The published protocol is \emph{harm-pure} if there exists a class label map $h:\classes\to\{0,1\}$ such that every variant in the same class has the same latent harm, i.e.,
    \[
        h^{\star}(v)=h(\classof{v}) \qquad \text{for all } v\in\variants.
    \]
\end{definition}

In the ideal case, we write $h(v)=h(\classof{v})$. Section~\ref{sec:uncertainty} relaxes this by separating audited class labels from latent variant-level harm and charging the mismatch to an explicit disagreement term.

\paragraph{Protocol-construction workflow.}
A practical audit can instantiate the protocol in four steps.
\begin{enumerate}[leftmargin=1.3em]
    \item Choose a bounded transformation family $T_0$ (for example spelling variants, caption edits, paraphrases, or interface-rendering changes) and the set of audit-relevant attributes that must remain fixed.
    \item Generate candidate pairs $(v,u)\in T_0$.
    \item Validate the pairs with a human study or another published rule, producing confidence scores $s(v,u)$ for ``same audited semantics'' and discarding disputed edges below the chosen threshold $\rho$.
    \item Publish the resulting graph, induced classes, and at least one sensitivity comparison under a stricter threshold or a pruned edge set.
\end{enumerate}
The paper's formal guarantees are conditional on this published protocol. The meta-game in which a platform tries to influence or negotiate the protocol is outside our core model; instead, we make the protocol explicit and require its sensitivity to be reported.

\paragraph{Worked public-data example.}
HateCheck is a public functional test suite for hate-speech detection, with 29 functionalities grouped into 11 classes and 3{,}728 validated cases \citep{Rottger2021-hc}. It is useful here not because it solves the audit problem, but because it makes protocol construction concrete. An auditor concerned with hateful exposure can fix a protected target, a directed/general flag, and a harm-relevant speech act, and then group paraphrastic realizations that preserve those attributes. \Cref{tab:hatecheckclass} shows one class built from four HateCheck cases.

\begin{table}[t]
    \caption{Worked semantic class from HateCheck's \texttt{derog\_neg\_emote\_h} functionality (cases 1, 15, 22, and 29).}
    \label{tab:hatecheckclass}
    \small
    \renewcommand{\arraystretch}{1.08}
    \begin{tabular}{@{}p{0.25\columnwidth}p{0.67\columnwidth}@{}}
        \toprule
        Audit-relevant semantics & General negative-emotion hate directed at the target group \emph{women}. \\
        Candidate variants & ``I hate women.''; ``I despise women.''; ``I detest women.''; ``I absolutely loathe women.'' \\
        Validation rule & Keep the cases in one class only if they retain the same gold label (hateful), target identity (women), and direction (general), and do not change other policy-relevant attributes such as threat or actionability. Disputed edges are excluded or reported separately in a sensitivity analysis. \\
        \bottomrule
    \end{tabular}
\end{table}

A platform strategy is a probability distribution $x\in\Strategies$ over the variants, where $\Strategies\subseteq\Delta(\variants)$ and $\Delta(\variants)=\{x\in\mathbb{R}_{\ge 0}^{\variants}:\sum_v x_v=1\}$ is the probability simplex. Throughout the core theorems, we take $\Strategies=\Delta(\variants)$, so the platform may concentrate any amount of recommendation mass on any individual variant; this is the worst case for the auditor and is what makes the manipulation-invariance property a meaningful guarantee. The invariance and certification results hold for any $\Strategies\subseteq\Delta(\variants)$ that is closed under within-class mass redistribution. Intuitively, $x_v$ is the fraction of recommendation opportunities in which the platform exposes variant $v$ to the audited population. The platform also receives utility $u(v)\geq 0$ from exposing $v$ (for example, an engagement or watch-time proxy).

\begin{definition}[Exposure, utility, and metric]
For a strategy $x\in\Strategies$,
\[
    \Harm(x)=\sum_{v\in\variants} x_v h(v),\qquad
    \Utility(x)=\sum_{v\in\variants} x_v u(v).
\]
An audit metric is induced by a score function $m:\variants\to [0,1]$:
\[
    \Metric_m(x)=\sum_{v\in\variants} x_v m(v).
\]
\end{definition}

We write $\Harm(x)$ for the latent harmful-exposure mass under ideal-case classwise harm labeling. When we relax harm-purity in \Cref{sec:imperfect}, we distinguish $\TrueHarm(x)=\sum_v x_v h^{\star}(v)$ (latent variant-level harm) from $\AudHarm(x)=\sum_v x_v \hat h(\classof{v})$ (the audit's classwise estimate); under harm-purity, these coincide.

The platform is also strategic. It can choose both \emph{what} to expose and \emph{which variant} of a semantic class to expose. The auditor accepts a deployment if the measured score is at most a budget $\tau$:
\[
    \Metric_m(x)\leq \tau.
\]
The platform then solves a utility-maximization problem subject to the audit budget. Because our later certificates quantify over all $x\in\Strategies$, the metric guarantees do not depend on any particular utility model; utility enters only when we want to select one adversarial best response in synthetic experiments.

This is a Stackelberg-style audit game in which the auditor publishes the protocol, scoring rule, and threshold, and the platform best-responds. We focus on measurement gaming by the platform and not on adversarial users or creators. Recent work has shown that user-side manipulation of risk-controlling recommenders is also a realistic threat \citep{De-Toni2026-xk}; our framework targets the platform itself.

\paragraph{Admissible manipulations.}
The platform may replace a variant $v$ with any other variant $v'$ with $\classof{v'}=\classof{v}$. This captures semantic-preserving rewrites, relabelings, caption changes, formatting changes, and multiple audit views that alter the measured score without changing the underlying harmfulness \emph{under the published protocol}. If the platform can generate new candidates from a known transformation family, the model extends by closing $\variants$ under that family. If it can change the underlying harmfulness itself or if the protocol is renegotiated, that is outside within-class manipulation and must be analyzed separately.

\section{Robustness Properties}
\label{sec:props-robustness}

\begin{definition}[Manipulation invariance]
A metric $\Metric_m$ is manipulation invariant if for all strategies $x,x'\in\Strategies$ with the same induced mass on semantic classes,
\[
    \sum_{v:\classof{v}=c} x_v = \sum_{v:\classof{v}=c} x'_v \quad\forall c\in\classes,
\]
we have $\Metric_m(x)=\Metric_m(x')$.
\end{definition}

Manipulation invariance states that the metric depends only on the semantically relevant content mix and not on the platform's choice of representation within a class.

\begin{definition}[Certification]
A metric $\Metric_m$ is $(\gamma,\beta)$-certifying over $\Strategies$ if
\[
    \Harm(x)\leq \gamma \Metric_m(x)+\beta
    \quad\text{for all }x\in\Strategies.
\]
\end{definition}

This definition is permissive on its own: any non-negative metric is trivially $(\gamma,1)$-certifying (since $\Harm(x)\le 1$), and the fragile metric is $(1/\delta,0)$-certifying with $\delta=\min_{c:h(c)=1}\min_{u\in c} m(u)$ whenever every harmful variant has a strictly positive score. The interesting question is whether the constants $\gamma,\beta$ come from a predeclared classwise structure that an auditor publishes \emph{before} the platform's strategy is observed. We capture this with a per-class coverage predicate.

\begin{definition}[$\varepsilon$-strict class-coverage certificate]
    \label{def:eps-strict}
    Let $\mathcal{P}=(P_c)_{c\in\classes}$ be a published per-class score profile with $P_c\in[0,1]$. A metric $\Metric_m$ admits an $\varepsilon$-strict class-coverage certificate via $\mathcal{P}$ at slack $\beta$ if
    \[
        \begin{aligned}
            &m(v)\ge P_{\classof{v}} \text{ for every } v\in\variants \text{ with } h(\classof{v})=1, \\
            &\min_{c:h(c)=1} P_c \ge \varepsilon,
        \end{aligned}
    \]
    and $\Metric_m$ is $(1/\varepsilon,\beta)$-certifying over $\Strategies$. The certificate is \emph{useful at audit budget $\tau$} if $\tau/\varepsilon+\beta<1$. When the certificate is read against latent rather than ideal harm via \Cref{cor:disagreement}, the slack $\beta=\bar\eta$ must itself be a per-variant disagreement bound (e.g., a Bernoulli error rate uniform over harmful-class variants); a class-average annotation agreement statistic is \emph{not} a defensible $\bar\eta$, because the platform can concentrate exposure on the small fraction of disputed variants---see the remark after \Cref{cor:disagreement}.
\end{definition}

Two aspects of this definition are formal rather than stylistic. First, the classwise lower bound is structural: the envelope satisfies $\Env(m)(v)\ge\alpha_c$ on every variant of every harmful class $c$ by construction, so the envelope's $\varepsilon$ equals the published coverage $\hat\alpha=\min_c\alpha_c$. The fragile metric satisfies this only for $P_c\le \min_{u\in c}m(u)$, which collapses $\varepsilon$ to the smallest manipulated-variant score in any harmful class. Second, the strictness threshold $\varepsilon$ is published before the platform moves and cannot be reverse-engineered from any particular adversarial strategy.

In particular, on the deterministic catalog of \Cref{tab:detcatalog}, the envelope admits a $0.85$-strict certificate (every harmful class has $\alpha_c\ge 0.85$), and at $\tau=0.20$ the certified ceiling $\tau/0.85=0.235$ is useful. The fragile metric admits at most a $0.10$-strict certificate (the smallest harmful-variant score is $0.10$ on $H2$\,\texttt{manip}); at $\tau=0.20$ this gives ceiling $\tau/0.10=2.0$, which is vacuous. The class-mean repair admits a $0.475$-strict certificate at $\tau=0.20$ giving ceiling $0.421$, useful but looser than the envelope. \Cref{cor:certificate} can therefore be read as showing that the envelope achieves the largest $\varepsilon$ available from the observed within-class maximum scores without inflating any class above an observed score, i.e., it is the strongest certificate among repairs that are classwise constant, pointwise conservative, and non-inflating beyond the observed class maximum. \Cref{cor:disagreement} extends this to $\beta=\bar\eta$ when harm-purity is relaxed. When we say below that the fragile metric \emph{violates the envelope-style class-coverage certificate}, we mean exactly this: it cannot match the envelope's $\varepsilon$ on the same published protocol within this comparison class. A useful audit requires both invariance to pure metric gaming and a useful $\varepsilon$-strict certificate.

\section{The Semantic-Envelope Repair}

\subsection{Fragility of direct scoring}

The most obvious audit metric scores each exposed variant directly using a classifier or rule-based detector. We call this the \emph{fragile metric}: $m_{\Frag}(v)=m(v)$.

\begin{proposition}[Direct variant scoring is manipulable]
\label{prop:fragile}
    Suppose there exists a semantic class $c$ with two variants $v_1,v_2\in c$ such that $m(v_1)>m(v_2)$. Then, $\Metric_{m}$ is not manipulation invariant.
\end{proposition}

\begin{proof}
    Let $x$ be the strategy that places unit mass on $v_1$, and let $x'$ place unit mass on $v_2$. Both $x$ and $x'$ assign the same total mass to class $c$ and zero mass to every other class; thus, they induce an identical class distribution and are therefore representation-equivalent in the sense of manipulation invariance. However, $\Metric_m(x)=m(v_1)$ and $\Metric_m(x')=m(v_2)$, and by hypothesis, $m(v_1)\neq m(v_2)$. Hence, $\Metric_m$ takes different values on representation-equivalent strategies, which is the definition of failing manipulation invariance.
\end{proof}

\Cref{prop:fragile} formalizes the core Goodhart failure mode: if the platform can change the measured score by swapping semantically equivalent variants, then an improved metric does not imply reduced harm.

\subsection{Semantic envelopes}

We repair this failure by lifting each variant's score to the maximum score of its semantic class.

\begin{definition}[Semantic-envelope lift]
For any score function $m:\variants\to[0,1]$, define
\[
    \Env(m)(v)=\max_{u\in\variants:\, \classof{u}=\classof{v}} m(u).
\]
The corresponding metric is $\Metric_{\Env(m)}$.
\end{definition}

The lift is conservative: it assumes that the auditor should score a class according to its worst measured representation among admissible variants.

\begin{theorem}[Minimal conservative invariant repair]
\label{thm:minimal}
    Let $\mathcal{G}$ be the set of score functions $\tilde m:\variants\to[0,1]$ satisfying:
    \begin{enumerate}[label=(\roman*),leftmargin=1.2em]
        \item \emph{classwise constancy}: $\tilde m(v)=\tilde m(u)$ whenever $\classof{v}=\classof{u}$; and
        \item \emph{conservativeness}: $\tilde m(v)\geq m(v)$ for all $v\in\variants$.
    \end{enumerate}
    Then $\Env(m)\in\mathcal{G}$, and for every $\tilde m\in\mathcal{G}$ and every $v\in\variants$,
    \[
        \Env(m)(v)\leq \tilde m(v).
    \]
    Hence, $\Env(m)$ is the unique pointwise minimum of $\mathcal{G}$ and therefore the least conservative repair in that comparison class.
\end{theorem}

\begin{proof}
    Class-wise constancy follows because the maximizing set depends only on the class. Conservativeness follows because $v$ belongs to its own class; therefore, $\Env(m)(v)\geq m(v)$. For minimality, let $\tilde m\in\mathcal{G}$ and fix a variant $v$. For every $u$ in the class of $v$, conservativeness yields $\tilde m(u)\geq m(u)$, whereas classwise constancy yields $\tilde m(u)=\tilde m(v)$. Therefore, $\tilde m(v)$ is an upper bound on $\{m(u):\classof{u}=\classof{v}\}$, implying $\tilde m(v)\geq \max_{u:\classof{u}=\classof{v}} m(u)=\Env(m)(v)$.
    \end{proof}
    
    \begin{remark}[Why max rather than mean or quantile?]
    Any classwise-constant aggregator restores manipulation invariance, and any such aggregator with strictly positive scores on harmful classes admits a certificate of the form $\Harm(x)\le (1/\alpha_g)\Metric_g(x)$, where $\alpha_g=\min_{c:h(c)=1} g(c)>0$. Therefore, mean and quantile repairs can also \emph{certify} harm, just with different (typically smaller) coverage constants. What distinguishes the semantic envelope is \emph{pointwise conservativeness}: it never lowers the score of a class relative to any already-observed variant. Mean or quantile repairs can recover utility; however, they can also down-score a class that contains a known high-risk representation, weakening the audit's conservative interpretation. \Cref{thm:minimal} therefore characterizes the envelope as the minimal repair in the safety-oriented (pointwise-conservative) comparison class, not as the only invariant repair that admits any certificate at all.
\end{remark}

\begin{proposition}[Manipulation invariance of semantic envelopes]
\label{prop:invariant}
    $\Metric_{\Env(m)}$ is manipulation invariant.
\end{proposition}

\begin{proof}
    Because $\Env(m)$ is a classwise constant, $\Metric_{\Env(m)}(x)$ depends only on the total mass that $x$ assigns to each class.
\end{proof}

\begin{proposition}[Monotonicity under protocol refinement]
\label{prop:refinement}
    Let $P_f$ and $P_c$ be two partitions of $\variants$ such that $P_f$ refines $P_c$. For any partition $P$, define
    \[
        \Env_{P}(m)(v)=\max_{u\in\variants:\,P(u)=P(v)} m(u).
    \]
    Then for every $v\in\variants$,
    \[
        \Env_{P_f}(m)(v)\leq \Env_{P_c}(m)(v).
    \]
\end{proposition}

\begin{proof}
    The cell of $v$ under $P_f$ is a subset of its cell under $P_c$; therefore, the maximum over the smaller set cannot exceed that over the larger set.
\end{proof}

\begin{remark}[Protocol sensitivity]
    If the auditor tightens validation, raises the acceptance threshold, or removes disputed equivalence edges, the induced partition refines, and the envelope weakly decreases pointwise by \Cref{prop:refinement}. Reporting results for a coarse partition $P_c$ and a fine partition $P_f\sqsubseteq P_c$ therefore yields a \emph{metric-sensitivity band}: a pair of envelope scores satisfying $\Env_{P_f}(m)(v)\le \Env_{P_c}(m)(v)$ for every $v$. To translate this into a harm-ceiling band, the auditor must also publish the corresponding $\hat\alpha$ and $\bar\eta$ for each protocol and instantiate \Cref{cor:disagreement}; refining the partition can lower the envelope score and simultaneously lower $\hat\alpha$; therefore, a tighter score does not automatically imply a tighter harm bound. \Cref{app:protocol-sensitivity} exhibits exactly this trade-off on the six-variant model with thresholds $\rho\in\{0.70,0.90\}$. This does not make incorrect merges harmless; if a published class mixes items with different latent harm labels, the purity assumption fails, and the disagreement term of \Cref{thm:disagreement} must absorb that error explicitly.
\end{remark}

\subsection{Certification}

Manipulation invariance alone is insufficient. We also want the score to upper-bound true harmful exposure.

\begin{theorem}[Class-stratified certification]
\label{thm:classwise}
    For a harmful class $c$ with $h(c)=1$, define its envelope score
    \[
        \alpha_c = \max_{u:\classof{u}=c} m(u).
    \]
    If $\alpha_c>0$, then every strategy $x\in\Strategies$ satisfies
    \[
        \sum_{v:\classof{v}=c} x_v
        \leq
        \frac{1}{\alpha_c}\sum_{v:\classof{v}=c} x_v \Env(m)(v).
    \]
    Consequently, for every set $S$ of harmful classes with $\alpha_c>0$ for all $c\in S$,
    \[
        \sum_{c\in S}\sum_{v:\classof{v}=c} x_v
        \leq
        \sum_{c\in S}\frac{1}{\alpha_c}\sum_{v:\classof{v}=c} x_v \Env(m)(v).
    \]
\end{theorem}

\begin{proof}
    Within a fixed harmful class $c$, every variant has an envelope score exactly $\alpha_c$. Hence
    \[
    \sum_{v:\classof{v}=c} x_v \Env(m)(v)=\alpha_c \sum_{v:\classof{v}=c} x_v,
    \]
    Dividing by $\alpha_c$ gives the first inequality as an equality. Summing over classes in $S$ gives the second statement, that is,
\end{proof}

\begin{corollary}[Global certification by lower-bounded harmful classes]
\label{cor:certificate}
    Let
    \[
        \alpha = \min_{c\in\classes:\, h(c)=1} \max_{u:\classof{u}=c} m(u).
    \]
    If $\alpha>0$, then for every strategy $x\in\Strategies$,
    \[
        \Harm(x)\leq \frac{1}{\alpha}\Metric_{\Env(m)}(x).
    \]
    Thus the semantic-envelope metric is $(1/\alpha,0)$-certifying in the ideal harm-pure case.
\end{corollary}

\begin{proof}
    Apply \Cref{thm:classwise} to every harmful class and use $\alpha_c\geq \alpha$.
\end{proof}

\begin{remark}[Interpreting $\alpha$ and the coverage prerequisite]
    The corollary is a coarse summary of a richer classwise profile. It is useful only when low-coverage harmful classes are absent or rare. Therefore, we treat the classwise coverage check as a prerequisite rather than a post-hoc safeguard: before invoking \Cref{cor:certificate}, the audit should publish the full distribution $\{\alpha_c\}_{c:h(c)=1}$, flag any class with $\alpha_c$ below a pre-declared detection threshold $\eta$, and separately flag any class with $\alpha_c=0$. A platform that can engineer a harmful class in which all admissible variants receive a vanishing score trivially defeats the aggregate bound for that class; no metric-level certificate can repair this without redesigning the measurement process. As a stylized illustration, suppose a single harmful class has $\alpha_c=0.05$---every admissible variant is scored below $0.05$ by the classifier---and suppose $\alpha_c$ sets the global $\alpha$. At a measured budget $\tau=0.20$, \Cref{cor:certificate} then gives $\Harm(x)\le \tau/\alpha=4$, which is vacuous for any strategy because $\Harm(x)\in[0,1]$. The protocol-level fix is structural: split the low-coverage class off the global certificate, report it as \emph{uncertified}, and either improve the classifier on that class or exclude it from the aggregate bound until coverage is restored. Averaging such a class into the global $\alpha$ silently degrades the certificate and is precisely the failure mode the prerequisite is meant to prevent. Section~\ref{sec:random} reports the empirical distribution of $\alpha$ across the synthetic catalogs.
\end{remark}

\subsection{Imperfect protocol construction and annotation}
\label{sec:imperfect}
\label{sec:uncertainty}

The preceding theorem family assumes that published classes are harm-pure and correctly labeled. To expose the cost of violating this assumption, we now separate the audited class labels from latent variant-level harm.

\begin{definition}[Audited labels and disagreement mass]
    Let $\hat h:\classes\to\{0,1\}$ be the audited class label published with the protocol, and let $h^{\star}:\variants\to\{0,1\}$ be the latent variant-level harm. Define the audited harmful exposure
    \[
        \AudHarm(x)=\sum_{v\in\variants} x_v \hat h(\classof{v}),
    \]
    the true harmful exposure
    \[
        \TrueHarm(x)=\sum_{v\in\variants} x_v h^{\star}(v),
    \]
    and the disagreement mass
    \[
        \Dmass(x)=\sum_{v\in\variants} x_v \bigl|h^{\star}(v)-\hat h(\classof{v})\bigr|.
    \]
\end{definition}

As written, $\Dmass(x)$ is the deterministic mass that the strategy $x$ places on variants whose latent harm disagrees with their audited class label; the stochastic-annotation reading, in which $\hat h(\classof{v})$ is a Bernoulli with class-conditional error $\epsilon_c$, follows by taking expectations and is treated explicitly in the remark after \Cref{cor:disagreement}. The disagreement term captures both annotation and protocol errors: if a class merges items whose latent harm differs, or if the audited class label is incorrect, the certificate degrades through $\Dmass(x)$ rather than silently failing.

\begin{theorem}[Certification with disagreement slack]
\label{thm:disagreement}
    Let
    \[
        \hat\alpha = \min_{c\in\classes:\,\hat h(c)=1}\max_{u:\classof{u}=c} m(u).
    \]
    If $\hat\alpha>0$, then every strategy $x\in\Strategies$ satisfies
    \[
        \TrueHarm(x)\leq \frac{1}{\hat\alpha}\Metric_{\Env(m)}(x)+\Dmass(x).
    \]
\end{theorem}

\begin{proof}
    For each variant $v$,
    \[
        h^{\star}(v)\leq \hat h(\classof{v}) + \bigl|h^{\star}(v)-\hat h(\classof{v})\bigr|.
    \]
    Multiplying by $x_v$ and summing over $v$ gives
    \[
        \TrueHarm(x)\leq \AudHarm(x)+\Dmass(x).
    \]
    Now every class with $\hat h(c)=1$ has envelope score at least $\hat\alpha$, so
    \begin{align*}
        \AudHarm(x)
        &= \sum_{c:\hat h(c)=1}\sum_{v:\classof{v}=c} x_v \\
        &\leq \frac{1}{\hat\alpha}\sum_{c:\hat h(c)=1}\sum_{v:\classof{v}=c} x_v \Env(m)(v) \\
        &\leq \frac{1}{\hat\alpha}\Metric_{\Env(m)}(x).
    \end{align*}
    The claim is proven by combining the two inequalities.
\end{proof}

\begin{corollary}[Bounded-disagreement certification]
\label{cor:disagreement}
    If the published protocol and annotation pipeline guarantee $\Dmass(x)\leq \bar\eta$ for all admissible strategies $x$, then
    \[
        \TrueHarm(x)\leq \frac{1}{\hat\alpha}\Metric_{\Env(m)}(x)+\bar\eta.
    \]
    Thus the semantic-envelope metric is $(1/\hat\alpha,\bar\eta)$-certifying with respect to latent harm.
\end{corollary}

\begin{remark}[From agreement studies to $\bar\eta$, with an adversarial caveat]
    The theorem isolates exactly where annotation uncertainty enters; however, the calibration must be adversarially valid because the platform can place an arbitrary mass on any single variant: $\sup_{x\in\Strategies}\Dmass(x)=1$ as soon as a single variant in any harmful class is audited-wrong; therefore, a \emph{class-average} agreement rate is not by itself a defensible $\bar\eta$. Two adversarially valid calibrations are available. In the deterministic case, the auditor validates \emph{every} variant in each harmful class to confirm $h^{\star}(v)=\hat h(\classof{v})$ pointwise; then, $\Dmass(x)=0$ on the audited deployment and \Cref{cor:disagreement} reduces to \Cref{cor:certificate}. In the stochastic case, the auditor reports a \emph{per-variant} Bernoulli error bound $\epsilon_c$ such that
    \[
        \Pr[h^{\star}(v)\neq \hat h(\classof{v})]\leq \epsilon_c \quad\text{for every } v\in\variants \text{ with } \classof{v}=c,
    \]
    then, writing $\mu_x(c)=\sum_{v:\classof{v}=c}x_v$,
    \[
        \mathbb{E}[\Dmass(x)]\leq \sum_{c\in\classes} \mu_x(c)\epsilon_c \leq \epsilon_{\max},
        \qquad
        \epsilon_{\max}=\max_{c\in\classes}\epsilon_c.
    \]
    A class-average annotation agreement rate (such as a HateCheck-style $\kappa=0.82$) is \emph{not} a substitute for per-variant $\epsilon_c$ in this bound because the platform can concentrate exposure on the small fraction of disputed variants and inflate $\Dmass(x)$ well past $1-\kappa$. An auditor who has \emph{only} a class-average $\kappa$ has two principled choices: (a) escalate to the deterministic case by validating every variant individually until each passes the harm-purity test, recovering the ideal-case certificate of \Cref{cor:certificate}; or (b) treat the class-average $\kappa$ as evidence that further per-variant annotation is required and refuse to publish a stochastic $\bar\eta$ until per-variant $\epsilon_c$ are measured. The protocol should publish whichever per-variant deterministic or stochastic disagreement bound it is prepared to defend in the worst case, not a class-average summary statistic.
\end{remark}

\begin{remark}[Deterministic, expected, and high-probability readings]
\label{rem:three-readings}
    \Cref{thm:disagreement} admits three distinct probabilistic readings, and the auditor must pick one before publishing. (i) \emph{Deterministic.} The disagreement mass $\Dmass(x)$ is taken as a fixed quantity bounded by $\bar\eta$ for every admissible $x$; the certificate $\TrueHarm(x)\le (1/\hat\alpha)\Metric_{\Env(m)}(x)+\bar\eta$ then holds as a worst-case inequality with no residual probability. This is the natural reading when the auditor has validated every harmful-class variant individually. (ii) \emph{Expected.} When the audit's class labels are stochastic with per-variant Bernoulli error $\epsilon_c$, the bound is most naturally stated as $\mathbb{E}[\TrueHarm(x)]\le (1/\hat\alpha)\Metric_{\Env(m)}(x)+\epsilon_{\max}$, with the expectation taken over the labeling randomness; this is the bound delivered by the chain $\mathbb{E}[\Dmass(x)]\le \epsilon_{\max}$. (iii) \emph{High-probability.} For a regulator that wants a confidence statement at level $1-\delta$, replace $\epsilon_c$ with a one-sided upper-confidence bound $\hat\epsilon_c(\delta)$ on the Bernoulli error from a finite annotation sample (Wilson, Clopper--Pearson, or empirical-Bernstein, depending on the annotation budget); $\bar\eta(\delta)=\max_c \hat\epsilon_c(\delta)$ then yields $\Pr[\TrueHarm(x)\le (1/\hat\alpha)\Metric_{\Env(m)}(x)+\bar\eta(\delta)]\ge 1-\delta$ uniformly over admissible $x$. The deterministic and stochastic versions are recovered as the $\delta\to 0$ and expectation-over-labeling specializations, respectively. The auditor should declare which reading the published $\bar\eta$ corresponds to, because reviewers downstream of the audit cannot recover this from the numerical value alone.
\end{remark}

\begin{remark}[Utility-model robustness]
    All invariance and certification results are quantified over \emph{all} strategies $x\in\Strategies$. Therefore, they survive arbitrary platform-side utility choices and utility misspecification: whatever strategy a platform selects, if it passes the semantic-envelope metric, the same certificate applies. Utility matters in the experiments only because we need a rule for choosing one adversarial best response from the feasible set.
\end{remark}

\section{Synthetic Stress Tests}

We now instantiate the framework using reproducible synthetic stress tests. The goal is not to model any deployed platform faithfully, but to expose the attack surface and the effect of the repair under controlled assumptions. To keep the role of each subsection legible, the section is organized around three layers: an \emph{illustrative stress-test} layer (deterministic catalog in \Cref{sec:detcatalog}, random catalogs in \Cref{sec:random}) that exhibits the attack surface and the size of the gaming gap; a \emph{mechanized consistency-checking} layer (finite-state grid in \Cref{sec:finite-state}, SMT in \Cref{sec:smt}, bounded MDP in \Cref{sec:prism}) that re-establishes the same invariants under three different abstractions; and a \emph{protocol-illustration} layer (HateCheck-derived protocol example in \Cref{sec:hatecheck}) that shows what a published semantic class looks like when grounded in a real benchmark. None of these layers is a real-platform empirical claim; the synthetic utility model is illustrative throughout, while the metric guarantees of \Cref{sec:props-robustness} remain valid for every feasible strategy regardless of utility. Both the deterministic and random catalogs assume harm-pure classes by construction; when a real protocol cannot justify that assumption, \Cref{thm:disagreement} is the appropriate certificate.

\paragraph{Reproducibility at a glance.}
Every number in this section is produced by scripts under \path{code/python/} and \path{experiments/prism/}, with outputs archived under \path{data/processed/}:
\begin{sloppypar}
    \begin{itemize}[leftmargin=1.3em,itemsep=0.1em]
        \item \path{reproduce.py} regenerates the deterministic and random-catalog tables and the $53{,}130$-state grid ($<\!1$\,min on a laptop).
        \item \path{verify_smt.py} emits and checks the $15$ Z3 queries and replays them in cvc5 ($<\!1$\,s total).
        \item \path{experiments/prism/audit_game.prism} is solved by PRISM-games~4.8 in the explicit engine ($<\!10$\,ms).
        \item \Cref{app:protocol-sensitivity} points to \path{tables/protocol_sensitivity.tex}, generated by the same script.
    \end{itemize}
\end{sloppypar}
The full artifact package, including a Dockerfile, a one-shot driver, and expected-output checksums, is described in Appendix~\ref{sec:open-science}.

\subsection{Deterministic case study}
\label{sec:detcatalog}

Our deterministic catalog contains four benign items and three harmful semantic classes, each with an original and a manipulated variant. Variants within a harmful class share the same true harm but differ in measured score and platform utility. This is a harm-pure stress test by construction: the only thing that changes within a class is the measurement view. The platform solves:
\[
    \max_{x\in\Strategies} \Utility(x)
    \quad \text{s.t.} \quad \Metric_m(x)\leq \tau.
\]
Because both utility and the audit constraint are linear in $x$, the optimum is obtained by linear programming. \Cref{tab:detcatalog} reports the full catalog used here; it has $|\variants|=10$ variants in $|\classes|=7$ classes (four benign single-variant classes and three harmful classes with an \texttt{orig}/\texttt{manip} pair each), $\hat\alpha=\min_{c:h(c)=1}\max_{u\in c} m(u)=0.85$, and a within-class envelope gap of at least $0.73$ on every harmful class.

\begin{table}[t]
    \centering
    \caption{Full deterministic catalog. Harmful classes pair an \texttt{orig} variant with a low-scoring \texttt{manip} variant of the same true harm. Envelope scores are the classwise maxima of the fragile scores.}
    \label{tab:detcatalog}
    \footnotesize
    \begin{tabular}{@{}llcccc@{}}
        \toprule
        Class & Variant & $h(\cdot)$ & $m_{\Frag}$ & $m_{\Env}$ & $u(\cdot)$ \\
        \midrule
        B1 & v0 & 0 & 0.05 & 0.05 & 0.55 \\
        B2 & v0 & 0 & 0.08 & 0.08 & 0.58 \\
        B3 & v0 & 0 & 0.10 & 0.10 & 0.62 \\
        B4 & v0 & 0 & 0.14 & 0.14 & 0.66 \\
        H1 & orig  & 1 & 0.95 & 0.95 & 0.78 \\
        H1 & manip & 1 & 0.15 & 0.95 & 0.82 \\
        H2 & orig  & 1 & 0.85 & 0.85 & 0.72 \\
        H2 & manip & 1 & 0.10 & 0.85 & 0.76 \\
        H3 & orig  & 1 & 0.90 & 0.90 & 0.70 \\
        H3 & manip & 1 & 0.12 & 0.90 & 0.74 \\
        \bottomrule
    \end{tabular}
\end{table}

\begin{figure}[t]
    \centering
    \includegraphics[width=\columnwidth]{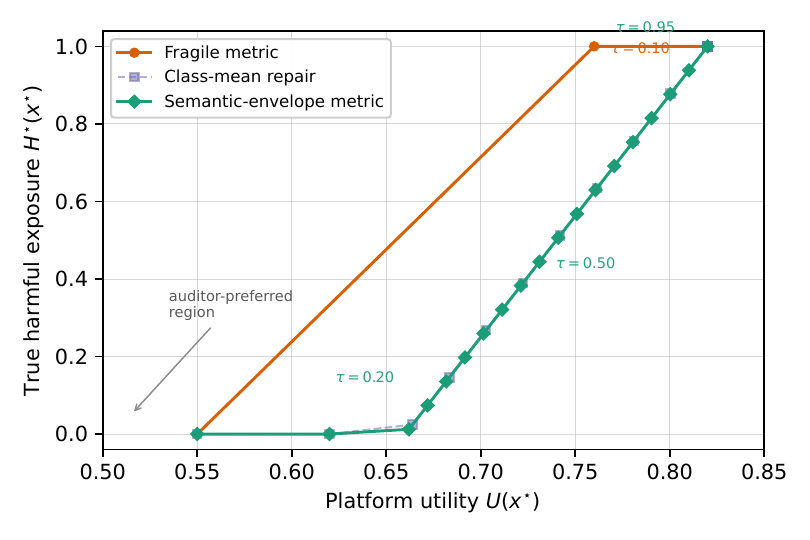}
    \caption{Utility--harm trajectories of the utility-maximizing platform strategy on the deterministic catalog, parameterized by the audit budget $\tau\in[0.05,0.95]$. Each curve traces $(\Utility(x^{\star}_\tau),\TrueHarm(x^{\star}_\tau))$ as $\tau$ grows. Rightward movement is what the platform wants (more utility); downward movement is what the auditor wants (less true harm), so trajectories that sit lower for the same horizontal position correspond to safer audits at any revealed utility level. The semantic-envelope trajectory (green diamonds) lies pointwise below the fragile trajectory (orange) in this auditor-preferred direction: at every platform utility the fragile metric attains, the envelope metric admits strictly less true harmful exposure. Equivalently, to reach the same true harm under the envelope metric the platform must give up additional utility, which is the measurement-induced safety tax. The class-mean repair (faded purple) restores manipulation invariance but sits slightly above the envelope because it is not one-sidedly conservative. Selected $\tau$ values are annotated.}
    \Description{Three trajectories of platform utility versus true harmful exposure. The fragile metric jumps from utility 0.55 at zero harm to utility 0.76 at full harm with one segment, then stays at full harm. The semantic-envelope and class-mean trajectories trace a near-coincident smooth diagonal from (0.55, 0) to (0.82, 1.0); both lie below the fragile trajectory at every shared utility, with class-mean slightly above envelope.}
    \label{fig:pareto}
\end{figure}

\Cref{fig:pareto} shows the main qualitative effect. At every platform utility level $\Utility(x^{\star})$ the fragile metric attains, the envelope metric admits strictly less true harmful exposure: the envelope trajectory dominates the fragile trajectory in the auditor-preferred direction. Concretely, the fragile metric reaches $\TrueHarm(x^{\star})=1.0$ at audit budget $\tau=0.10$ (utility $0.760$) and stays there for every larger $\tau$, whereas the envelope metric admits $\TrueHarm(x^{\star})=0.000$ at $\tau=0.10$, $0.012$ at $\tau=0.15$, and $0.074$ at $\tau=0.20$. The class-mean repair tracks the envelope closely (it restores manipulation invariance) but admits slightly more true harm at every $\tau$ because it is not one-sidedly conservative.

\begin{table}[t]
    \centering
    \caption{Deterministic case study under selected audit budgets.}
    \label{tab:case}
        \begin{tabular}{@{}l S[table-format=1.2] S[table-format=1.3] S[table-format=1.3]@{}}
        \toprule
        Metric and budget & {Utility} & {True harm} & {Measured score} \\
        \midrule
        Fragile @ $0.10$ & 0.760 & 1.000 & 0.100 \\
        Envelope @ $0.10$ & 0.620 & 0.000 & 0.100 \\
        Fragile @ $0.15$ & 0.820 & 1.000 & 0.150 \\
        Envelope @ $0.15$ & 0.662 & 0.012 & 0.150 \\
        Fragile @ $0.20$ & 0.820 & 1.000 & 0.150 \\
        Envelope @ $0.20$ & 0.672 & 0.074 & 0.200 \\
        \bottomrule
    \end{tabular}
\end{table}

\Cref{tab:case} illustrates the security meaning of the repair. At identical measured budgets, the fragile metric can accept a deployment whose actual harmful exposure is maximal, whereas the semantic-envelope metric forces the platform to substitute toward benign content.

To separate invariance from conservativeness, we also evaluated the non-conservative class-mean repair
\[
    \Mean(m)(v)=\frac{1}{|\{u:\classof{u}=\classof{v}\}|}\sum_{u:\classof{u}=\classof{v}} m(u).
\]
At budget $0.20$, the class-mean repair yields utility $0.683$ and true harmful exposure $0.146$, compared with $0.672$ and $0.074$ for the semantic envelope, and $0.820$ and $1.000$ for the fragile metric. \Cref{tab:repair-compare} summarizes the qualitative positions of the three repairs on the properties this paper cares about. The fragile metric fails invariance (\Cref{prop:fragile}); both classwise-constant repairs restore it; both also admit a class-stratified certificate, since either one yields a positive minimum-coverage constant (here $\alpha_{\text{mean}}=0.475$ versus $\alpha_{\text{env}}=0.85$). What distinguishes the envelope is pointwise conservativeness: it never down-scores a class containing a high-risk variant, which the auditor needs to defend the certificate's worst-case interpretation. The comparison makes the design choice explicit: once pointwise conservativeness is relaxed, utility can be recovered, but the certificate's coverage constant shrinks and the audit's worst-case reading weakens.

\begin{table}[t]
    \centering
    \caption{Qualitative comparison of the three audit-metric repairs against the security properties of Section~\ref{sec:props-robustness}. Utility and true-harm columns are evaluated on the deterministic catalog at budget $\tau=0.20$.}
    \label{tab:repair-compare}
    \small
    \renewcommand{\arraystretch}{1.0}
    \begin{tabularx}{\linewidth}{@{}lccXcc@{}}
        \toprule
        Repair & Invariant? & Conservative? & Certifies? & Utility & True harm \\
        \midrule
        Fragile    & $\times$ & $\times$ & only vacuously, at the smallest harmful-variant score & $0.820$ & $1.000$ \\
        Class-mean & $\checkmark$ & one-sided fails & yes, but with weaker coverage & $0.683$ & $0.146$ \\
        Envelope   & $\checkmark$ & $\checkmark$ & yes, strongest non-inflating coverage (Cor.~\ref{cor:certificate}) & $0.672$ & $0.074$ \\
        \bottomrule
    \end{tabularx}
\end{table}

\subsection{Finite-state verification}
\label{sec:finite-state}

To verify the certification claim on a finite model, we construct a smaller six-variant instance and exhaustively enumerate all mixed strategies on a probability grid of step size $0.05$ on the simplex (each $x_v$ is a multiple of $0.05$ summing to $1$). The number of such strategies is the number of weak compositions of $1/0.05=20$ into $|\variants|=6$ parts, $\binom{20+6-1}{6-1}=\binom{25}{5}=53{,}130$.

Let $\alpha$ be the minimum envelope score among harmful classes in the small model; here, $\alpha=0.85$. We compute the maximum value of
\[
    \Harm(x)-\Metric_m(x)/\alpha
\]
over all enumerated strategies. For the fragile metric, the maximum violation is $0.882$, which is achieved by concentrating all the mass on a low-score harmful variant. For the semantic-envelope metric, the maximum violation is $0$, matching \Cref{cor:certificate}. This finite-state check is not required for the theorem, but it demonstrates how a reviewer or regulator could mechanically verify the property on a bounded instance.

\subsection{Symbolic consistency checks}
\label{sec:smt}

Enumeration certifies the properties of a bounded discretization. To check the implementation of the full continuous simplex of mixed strategies for concrete catalog instances, we also encode the metric properties as satisfiability queries in the Z3 SMT solver \citep{DeMoura2008-z3}. We treat these SMT results as \emph{mechanized consistency checks} for the catalog encodings, not as replacements for the pen-and-paper proofs. The implementation uses Z3~4.16.0 in logic \texttt{QF\_LRA} with \texttt{random\_seed=0}, a per-query timeout of 10\,s, and a strict-inequality tolerance of $\epsilon=10^{-8}$ in the fragility witness. The artifact archives the emitted SMT-LIB2 instances under \path{data/processed/smt_queries/}, so the exact queries can be rerun or dropped into another solver.

Mathematical encoding uses only rational constants, linear equalities, and inequalities over the reals. For every catalog instance, we introduce variables $x_v\in\mathbb{R}_{\geq 0}$ such that $\sum_v x_v = 1$ and check three queries:

\begin{enumerate}[leftmargin=1.4em]
    \item \emph{Envelope invariance.} Seek two strategies, $x, y$ with identical class masses, $\sum_{v:\classof{v}=c}x_v = \sum_{v:\classof{v}=c}y_v$ for every $c\in\classes$, but distinct envelope metrics. UNSAT is consistent with \Cref{prop:invariant}.
    \item \emph{Fragility witness.} Seek two such strategies that differ in the fragile metric. SAT exhibits manipulation according to \Cref{prop:fragile}.
    \item \emph{Certification.} Seek any strategy on the simplex with $\alpha\cdot\Harm(x) > \Metric_{\Env(m)}(x)$, where $\alpha = \min_{c:h(c)=1}\alpha_c$. UNSAT is consistent with \Cref{cor:certificate} on the tested instance.
\end{enumerate}

\noindent We run the three queries on the paper's six-variant model and on a family of synthetic catalogs with $n\in\{15,25,45,85\}$ variants (five benign classes and $n_h\in\{5,10,20,40\}$ harmful classes, same generator as \Cref{sec:random}). All envelope-invariance and certification queries return UNSAT (no counterexample exists in the encoded simplex); all fragility queries return SAT (the solver returns an explicit manipulation witness $x,y$). Solver time is below $63$ms per query on all instances, and below $0.11$s total across the five instances and three properties.

A separate scalability sweep (Z3 only) extends the same three queries to random catalogs with $n_h\in\{100,250,500\}$ harmful classes ($n\in\{205,505,1005\}$ total variants). The verdicts remain unchanged from the small-instance regime: envelope invariance is UNSAT and certification is UNSAT in every case, and fragile manipulation is SAT with an explicit witness. The slowest single query at $n=1005$ completes in under $3$\,s on commodity hardware, and the full sweep of nine queries finishes in under $5$\,s; therefore, the encoding remains tractable two orders of magnitude past the toy regime; per-query timings are persisted in \path{data/processed/scalability.csv} and \path{scalability_summary.json}.

To guard against solver-specific encoding artifacts, we replay every query in cvc5 as an independent engine. The exported \texttt{QF\_LRA} instances are parsed by cvc5 under the same 10\,s per-query timeout, and the resulting verdicts are recorded in \path{data/processed/smt_crosssolve.csv}. Across all $15$ queries (five catalog instances $\times$ three properties), Z3 and cvc5 return the same verdict: every envelope-invariance query is UNSAT in both solvers, every fragility query is SAT in both, and every certification query is UNSAT in both. Cross-solver agreement is $15/15$. Therefore, we treat the SMT stage as a mechanized consistency check that is not tied to a single solver implementation, while keeping the pen-and-paper proofs of \Cref{prop:invariant,prop:fragile,cor:certificate} as the primary source of correctness.

\subsection{Temporal verification via bounded MDP}
\label{sec:prism}

Enumeration and SMT certify the static properties of the audit score. An orthogonal question is whether the same qualitative effect survives in a \emph{sequentialized} toy audit. To answer this, we encode the audit as a small Markov decision process in the PRISM-games~4.8 probabilistic model checker \citep{Kwiatkowska2020-prismgames} and verify reachability rewards in PCTL via the $\mathtt{R}\{\cdot\}\mathtt{max}{=?}\,[\,\mathtt{F}\,\cdot\,]$ operator, where $\mathrm{R}_{\max}$ quantifies over deterministic platform strategies and returns the maximum expected one-shot transition reward over reachable terminal states. This is a proof-of-concept temporal model and not a realistic regulatory audit protocol.

The model reuses the canonical two-variant instance (one benign class at a cost $0.10$; one harmful class with an original variant at a cost $0.90$ and a manipulated variant at a cost $0.10$; envelope cost $0.90$; audit budget $\tau=0.20$). We instantiate two separate metric settings, \texttt{fragile} and \texttt{envelope}; the auditor publishes the metric; therefore, it is not a platform choice. Within each metric setting, the platform makes two nondeterministic choices over successive rounds: (i) a mass allocation $m_H\in\{i/20 : 0\le i\le 20\}$ on a grid of step $0.05$, and (ii) a variant assignment \texttt{orig} or \texttt{manip}. A one-shot transition reward $\Harm(x)=m_H$ fires exactly when the audit passes under the chosen cell. PRISM's $\mathrm{R}_{\max}$ operator then quantifies over every deterministic platform strategy and returns the worst-case true-harm exposure in each $(\text{metric},\text{variant})$ combination:
\[
    \varphi_{m,v} \;=\; \mathtt{R}\{\text{``harm}_{m,v}\text{''}\}\mathtt{max}{=?}\,[\,\mathtt{F}\;\text{``done''}\,].
\]
Every path reaches \texttt{done} in three rounds, so each expected reward is finite.

\begin{table}[t]
    \centering
    \caption{Adversary-optimal true-harm exposure $\varphi_{m,v}$ computed by PRISM-games over the 3-round audit MDP (148 reachable states, 231 transitions, explicit engine, $<10$ms per query). Under the fragile metric the two variant cells diverge by $0.90$, exhibiting manipulation; under the envelope metric they coincide, matching \Cref{prop:invariant}.}
    \label{tab:prism}
    \small
    \begin{tabular}{lcc}
        \toprule
        Metric & $v=\text{orig}$ & $v=\text{manip}$ \\
        \midrule
        Fragile  & $0.10$ & $1.00$ \\
        Envelope & $0.10$ & $0.10$ \\
        \bottomrule
    \end{tabular}
\end{table}

\Cref{tab:prism} reports the four values returned by PRISM-games. The envelope row is constant across variants, so the invariance property $\varphi_{\Env,\texttt{orig}}=\varphi_{\Env,\texttt{manip}}$ holds as an \emph{operational} temporal claim, not only as an algebraic one. The fragile row splits: swapping the variant at fixed class mass raises the adversary's optimum from $0.10$ to $1.00$. The gap $\varphi_{\Frag,\texttt{manip}}-\varphi_{\Frag,\texttt{orig}}=0.90$ is of the same magnitude as the mean gap of $0.884$ from the random-catalog experiment in \Cref{sec:random}, and the envelope row certifies it away in a way an auditor can mechanically reproduce. The artifact ships the PRISM model at \path{experiments/prism/audit_game.prism} and the property file at \path{experiments/prism/audit_game.props}; the full solver trace is archived in \path{data/processed/prism_results.txt}.

\subsection{Random-catalog experiment}
\label{sec:random}

Next, we stress-test how often the gaming gap appears beyond a single hand-crafted example. We sample 500 random catalogs with five benign classes and five harmful classes. Each harmful class has an original high-score variant and a manipulated low-score variant with slightly improved utility. As in the deterministic case, the experiment assumes harm-purity by construction: the manipulation changes only the audit score, not latent harm. For each catalog, we compare the utility-maximizing strategy under the fragile metric to the utility-maximizing strategy under the semantic-envelope metric with the same audit budget, $\tau=0.20$.

\begin{figure}[t]
    \centering
    \includegraphics[width=\columnwidth]{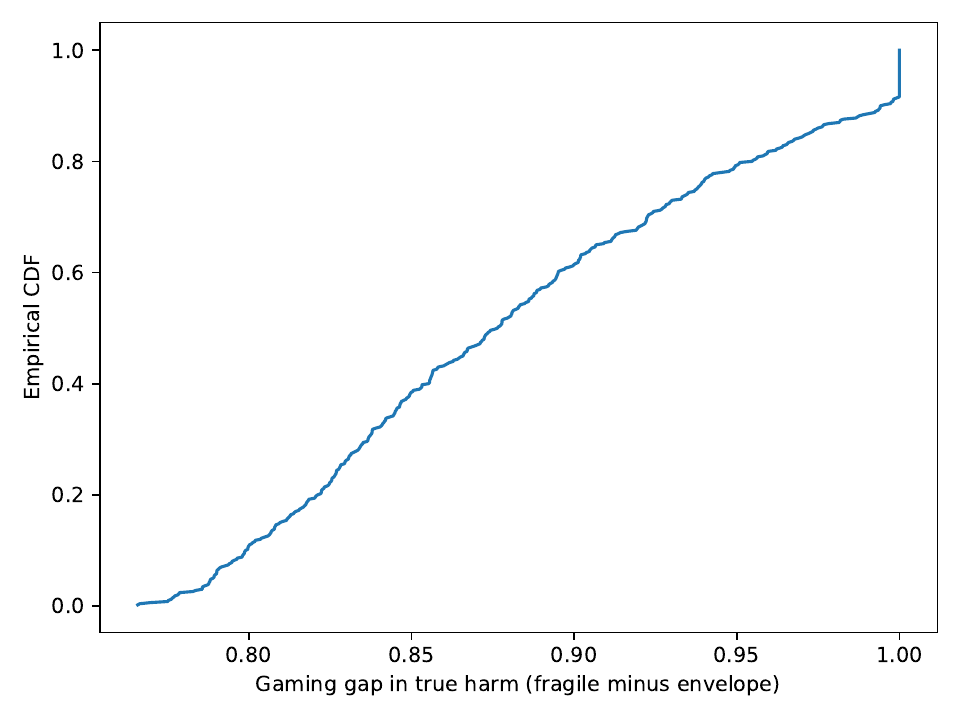}
    \caption{Empirical CDF of the gaming gap in true harmful exposure (fragile optimum minus semantic-envelope optimum) across 500 random catalogs at budget $0.20$. The gap exceeds $0.76$ in every instance.}
    \Description{An empirical cumulative distribution of the harm gap across 500 random catalogs. All observed gaps are large, beginning above 0.76 and extending to 1.00.}
    \label{fig:cdf}
\end{figure}

The results again show a large and systematic gap. The fragile-metric optimum has a mean true harmful exposure $1.00$, whereas the semantic-envelope optimum has a mean true harmful exposure $0.116$. The mean gap is $0.884$, and the fragile optimum has a strictly larger harmful exposure in all 500 instances. Platform utility is higher under the fragile metric (mean $0.887$ versus $0.684$), which is exactly why the metric is strategically attractive to the platform.

The class-mean repair again behaves intermediately: across the same 500 catalogs, it yields a mean true harmful exposure $0.266$ and a mean utility $0.712$. This confirms the role of conservativeness from \Cref{thm:minimal}: weaker classwise aggregators can interpolate toward higher utility, but they do so by admitting more harmful mass at the same budget.

\Cref{fig:cdf} shows the distribution of the harm gap. Every sampled instance exhibits a substantial gap: even the minimum observed gap is above $0.76$. In this synthetic family, the envelope coverage values are intentionally high by construction (global $\alpha$ mean $0.789$, minimum $0.750$, and no instance has a harmful class with $\alpha_c<0.20$), so the global corollary is informative on every instance. We report this prevalence deliberately: the coverage prerequisite of \Cref{cor:certificate} can fail, and the share of real-world catalogs with low-coverage harmful classes is an \emph{empirical} question an auditor must answer with data from the deployed domain. Real audits should not assume this harm-pure, high-coverage regime; they should report the full classwise profile from \Cref{thm:classwise} along with the share of harmful classes falling below $\eta$. \Cref{app:sensitivity} traces the certified ceiling on $\TrueHarm(x)$ at this audit budget across the realistic ranges $\hat\alpha\in[0.20,1.00]$ and $\bar\eta\in\{0,0.05,0.10,0.15,0.20\}$, so an auditor who has measured even crude bounds on coverage and disagreement can read off the resulting certificate without rerunning the LP. Nor should this experiment be read as a scalability claim for live audits: it shows that the attack surface is not unique to one hand-crafted catalog, but it does not replace a pilot on a deployed system or a post-hoc reanalysis of a published audit.

\subsection{HateCheck-derived protocol example with synthetic scores}
\label{sec:hatecheck}
\label{sec:hatecheck-mini}

The random-catalog experiment confirms that the fragile-envelope gap is not an artifact of a single hand-crafted example; it still draws classes from a synthetic generator. To move one step closer to a real protocol, we instantiate the framework on class structure taken directly from HateCheck \citep{Rottger2021-hc}: three harmful classes drawn from the \texttt{derog\_neg\_emote\_h} and \texttt{threat\_dir\_h} functionalities with distinct targets, and one benign control drawn from \texttt{ident\_neutral\_nh}. Each harmful class contains three HateCheck cases as within-class variants. Because we do not have a deployed classifier or an annotation pipeline for this example, we use a \emph{synthetic-but-rule-based} scoring function: within a harmful class, the first case receives score $0.92$ (the ``obvious'' variant a classifier flags) and subsequent cases receive $0.92\cdot(1-0.30i)$, giving a plausible descending profile for paraphrases or distractors that a real classifier might miss. Utility is the inverse, and therefore, the platform prefers the lower-scored variants. Benign variants receive a uniform low score and moderate utility. We frame this as an \emph{illustrative verification artifact}, not an evaluation of any real hate-speech classifier or audit pipeline: only the class structure and variant labels are real; scores and utilities are design choices locked before running the LPs.

\input{tables/hatecheck_mini.tex}

At budget $\tau=0.20$, the qualitative pattern from the synthetic experiments transfers as follows: under the fragile metric, the utility-maximizing platform routes substantial mass through low-score harmful variants in \texttt{threat\_dir\_h}, whereas under the semantic-envelope metric, the same budget forces most of that mass onto the benign class. The other two harmful classes contribute no mass in the LP optimum because, at the chosen budget and utility profile, their score-to-utility trade-off is dominated by the selected threat and benign classes (the LP allows fractional mass on any variant, so the absence is a dominance result, not a feasibility obstruction). Because scores are rule-based, this instantiation can be interpreted as a working HateCheck protocol instance illustrating that the attack surface and repair survive a protocol whose classes come from a public benchmark; it is not a measured vulnerability of any deployed classifier.

\section{Discussion}

\paragraph{Why the repair works.}
The semantic-envelope metric removes one specific attack surface: the platform can no longer gain audit slack by swapping among semantically equivalent variants. In the audit game, that forces the platform to reduce mass on harmful classes if it wants to meet a tighter budget. Therefore, repair is best understood as a \emph{measurement-hardening step}, not as a complete audit framework.

\paragraph{Where this layer sits in the audit-robustness stack.}
A regulatory audit pipeline has at least six attack surfaces under platform-side or user-side strategies: (i) within-class representation choice (this paper), (ii) protocol negotiation and threshold lobbying (sketched above), (iii) classifier evasion against the underlying detector, (iv) user-side feedback manipulation against risk-controlling recommenders \citep{De-Toni2026-xk}, (v) sampling error in the audit instrument, and (vi) cross-platform substitution. Each layer admits its own formal specification, its own threat model, and, once specified, its own verification or measurement question. We address layer~(i) because it is where a published audit metric becomes its own attack surface and where a small formal repair (\Cref{thm:minimal}) restores invariance without touching the protocol, the classifier, or the user side. This paper fixes the \emph{measurement specification} of the audit. Once layer~(i) is pinned, layers~(ii)--(vi) can be analyzed against a stable measurement ground truth instead of through an unspecified metric. \Cref{app:operational-harm-purity} gives the deployment-validation procedure that connects the formal specification to a real audit.

\paragraph{Protocol design is part of the audit.}
The paper no longer treats semantic classes as a hidden oracle. A published protocol now consists of a bounded transformation family, an attribute-preservation checklist, a validation study, and a threshold. Tightening the protocol refines the partition and weakly lowers the envelope by \Cref{prop:refinement}, so auditors can report a monotone sensitivity band rather than a single unexplained partition. What this does \emph{not} solve is protocol negotiation or strategic lobbying over the protocol itself; that higher-level meta-game remains outside scope.

\paragraph{A toy protocol-negotiation meta-game.}
To sketch where the envelope remains useful and where it breaks under protocol-level strategy, consider a two-period variant of the audit game. In period~$0$, the platform can propose a transformation family $T_0'\supseteq T_0$ and lobby for a threshold $\rho'\le\rho$; the regulator admits a subset of proposals after validation. In period~$1$, within-class manipulation proceeds as in the main model against the accepted protocol. Two observations follow directly from \Cref{prop:refinement}. First, any admitted edge that genuinely preserves audited semantics coarsens the partition and weakly \emph{raises} the envelope score: the platform gains nothing from lobbying for such edges because the envelope already covers them. Second, admitting an edge that merges items with different latent harm labels breaks harm-purity; the disagreement mass $\Dmass(x)$ of \Cref{thm:disagreement} absorbs the error, but only up to the slack $\bar\eta$. The regulator's best response is therefore to tie threshold $\rho$ and edge admission to a published validation target and to publish a \emph{refined} envelope under the stricter protocol alongside the permissive one. A fully game-theoretic analysis of this meta-game, including platform incentives to inflate $T_0'$ and regulator commitment to $\rho$, is outside the scope of this paper, but the envelope provides a checkable quantity at every period of such a game.

\paragraph{Annotation uncertainty enters as slack, not as a footnote.}
The disagreement theorem separates audited class labels from latent harm and charges the mismatch to $\Dmass(x)$. This is the appropriate place for annotation uncertainty, class impurity, and false equivalence edges to appear. If a validation study can upper-bound disagreement mass by $\bar\eta$, then the final certificate degrades additively by exactly that amount. If it cannot, the paper should not claim to have an exact certificate.

\paragraph{Novelty boundary.}
The closest conceptual ancestors are strategic classification and reward hacking \citep{Hardt2015-ci,Skalse2022-oz}. We do not claim a fundamentally new equilibrium concept beyond those literatures. The novelty here is the object of analysis---published platform-audit metrics---together with a concrete protocol model, a simple max-over-class repair with a minimality theorem, and a reviewer-checkable verification artifact for this threat model.

\paragraph{What the synthetic results do and do not show.}
The synthetic stress tests show that within-class manipulation can create large gaps even in small catalogs, and that the envelope removes this gap in the tested family. The certificate's quantitative bound is utility-agnostic by construction: \Cref{cor:disagreement} quantifies over every admissible strategy; therefore, the LP-selected best response in \Cref{tab:repair-compare} is illustrative of \emph{which} strategy a utility-maximizing platform would pick, not of \emph{whether} the bound holds. The synthetic results, by themselves, do not establish platform-side acceptance of the utility cost on real engagement data, the empirical distribution of class-coverage $\hat\alpha$ and disagreement mass $\bar\eta$ on a deployed audit, the scalability of the SMT and PRISM stages beyond the toy regime in \Cref{sec:smt,sec:prism}, sampling-error control, or regulator-platform bargaining over the protocol itself. Each of these is a separate, self-contained empirical question; the operational harm-purity workflow of \Cref{app:operational-harm-purity} prescribes the measurement steps an auditor would take to answer the coverage and disagreement parts, and the protocol-negotiation sketch above isolates the bargaining layer. The natural next steps are therefore concrete: a post-hoc reanalysis of a published audit to populate $\hat\alpha$ and $\bar\eta$ from real protocol data, and a multi-round game-theoretic extension that nests the protocol-negotiation meta-game inside the within-class certificate.

\paragraph{Relationship to formal verification.}
Our proofs characterize universal properties over all admissible strategies, and the enumeration, SMT, and bounded-MDP experiments (\Cref{sec:finite-state,sec:smt,sec:prism}) show how the same properties can be checked mechanically at three different levels of abstraction. The PRISM-games stage is technically a single-player MDP rather than a two-player stochastic game; we use the PRISM-games engine because it natively supports the reachability-reward operators that are required. In this version, we present the SMT stage as an implementation-level consistency check rather than as the sole proof source, and we document the exact Z3 configuration together with the emitted SMT-LIB queries and the cvc5 cross-solver replay for artifact-time reproduction. A natural next step is the development of an interactive theorem prover for end-to-end machine-checkable proofs.

\section{Related Work}

\textbf{Strategic manipulation and reward hacking.}
Our work is closest to strategic classification and reward-hacking theory. Strategic classification studies agents who manipulate features to receive favorable outcomes \citep{Hardt2015-ci}. Reward-hacking work asks when optimizing a proxy objective harms the true objective and when a proxy can be called robust or unhackable \citep{Skalse2022-oz}. We borrow this adversarial viewpoint. The distinction here is the object under attack: a published platform-audit metric, not a decision boundary or a learned reward model. Strategic classification typically asks how an agent crosses a classifier boundary by moving in feature space; our platform instead chooses which measurement view of fixed harmful content is exposed under a published audit protocol. The semantic-class protocol and envelope repair make this distinction operational.

\textbf{Performative prediction and endogenous distributions.}
\citet{Perdomo2020-oa} study predictors whose deployment changes the future data distribution and define performative stability as an equilibrium concept for retraining. Our setting is also endogenous, but the mechanism is different: the platform strategically reallocates exposure across measurement views of fixed harmfulness classes rather than subjects responding to a deployed predictor. This lets us prove invariance and certification results about \emph{metrics} rather than convergence or stability of a learning process.

\textbf{Security measurement and evaluation fragility.}
Recent CCS work shows that evaluation protocols can fail under stronger adversaries. \citet{Aerni2024-sm} demonstrate that empirical privacy-defense evaluations are misleading when they ignore adaptive attacks or vulnerable samples. \citet{Wang2025-gt} show that membership-inference attacks used as privacy tools differ substantially in reliability and coverage. Our setting is analogous: a safety metric is trustworthy only if it survives strategic pressure from the system being audited.

\textbf{Auditing and safety in recommender systems.}
\citet{Messmer2023-ar} argue that DSA-style audits require concrete, risk-scenario-based procedures rather than high-level transparency promises. \citet{Sharma2024-yp} provide a causal recipe for defining auditing metrics in recommenders. Harm-aware recommendation work shows that optimizing engagement under user dynamics can amplify harmful pathways and that naive mitigation may fail \citep{Chee2024-nl,Ribeiro2023-sl}. We address a different layer: whether the \emph{audit metric itself} remains sound under platform-side manipulation, given the protocol that defines it.

\textbf{Audit access and audit measurement.}
A complementary line of work focuses on the \emph{access} side of platform audits. \citet{Burnat2026-xq} catalogue ``audit blindspots'' arising from platform API restrictions under the EU Digital Services Act, where the data needed to audit content moderation and algorithmic amplification is itself foreclosed. Our setting assumes the auditor can obtain the requested signal and asks a different question: even with full access, the within-class representation choice on which a published metric depends can be strategically manipulated, so measurement design is its own audit-hardening surface.

\textbf{Functional testing and provable risk control.}
HateCheck introduces a public functional-test suite that exposes brittle keyword dependence and poor behavior on contrastive cases in hate-speech detection models \citep{Rottger2021-hc}. We use it differently: as a concrete source of protocol-level semantic classes. Recent work on conformal risk control has proposed recommender mechanisms with formal safety guarantees \citep{De-Toni2025-bc}. Follow-up work has shown that such mechanisms can still be vulnerable to coordinated manipulation through feedback channels \citep{De-Toni2026-xk}. Our model is complementary: rather than manipulating training or feedback, the platform manipulates the audit measurement channel.

\textbf{Distinction from conformal prediction.}
A natural question is whether conformal prediction (CP) provides the certificate that this study constructs. CP provides distribution-free coverage guarantees by calibrating prediction sets on an exchangeable holdout, so that the population-level miscoverage rate matches a published $\alpha$ with high probability. The settings differ in two structural ways. First, CP's guarantees rely on exchangeability between the calibration set and the deployed instances. Under our threat model, the platform observes the published protocol and re-routes exposure mass toward semantically equivalent low-score variants \emph{after} the protocol is fixed; the deployed marginal therefore drifts away from the calibration marginal in a direction the platform chooses. CP's coverage guarantee degrades silently under this kind of strategic distribution shift unless paired with adversary-aware recalibration. Second, CP is score-function agnostic by design: it wraps any scoring rule with a calibrated rejection threshold and certifies the rule's miscoverage rate, not the rule's robustness to within-class representation choice. Our certificate is metric-specific: it asserts that a particular published score function (the semantic envelope) admits an $\varepsilon$-strict class-coverage certificate and that direct variant scoring does not. The two methods address different layers and can be combined: CP for distribution-free uncertainty quantification on individual scores and the envelope for the audit-level metric. The joint construction is future work.

\section{Conclusion}

We treat a published platform-safety metric as a security object. The two requirements are that no platform strategy can improve the score through semantically irrelevant manipulations and that the score upper-bounds true harmful exposure under a published per-class coverage profile. Direct variant scoring fails the first requirement. The semantic-envelope lift restores manipulation invariance, is the least conservative classwise repair in its comparison class, and yields a simple global certificate when harmful classes are sufficiently covered.

The paper makes three further points explicit. First, semantic classes are part of the published audit protocol, not an oracle: the protocol can be varied and its sensitivity reported. Second, annotation and protocol error do not disappear by assumption; they enter the final guarantee as disagreement slack. Third, the theorem is utility-agnostic, whereas the LP experiments simply instantiate one utility-maximizing best response under a stated utility model. The broader lesson is methodological. Platform audits should not stop at reporting descriptive metrics; they should analyze the metric under the behavior of an optimizing adversary and clearly state what uncertainty remains outside the certificate.

The framework is a measurement-hardening component and not a complete regulatory-audit solution. The latent-harm version of the certificate (\Cref{cor:disagreement}) is only as tight as the published $\bar\eta$ is honest: the bound holds for every platform strategy, but it certifies useful harm ceilings only when $\bar\eta$ is calibrated from per-variant disagreement evidence rather than a class-average summary. The next steps are empirical grounding on a real audit (to populate $\hat\alpha$ and per-variant $\bar\eta$ from deployed protocol data) and richer formal models that nest the protocol-negotiation meta-game inside the within-class certificate.

\singlespacing
\printbibliography

\clearpage

\appendix
\section{Ethical Considerations}

This study uses synthetic catalogs, a worked example derived from a public benchmark, and no live-platform intervention. It does not involve human subjects, personal data, or operational interactions with a deployed service. The public benchmark example reproduces only a few short hate-speech test cases from HateCheck to illustrate semantic-class construction; these excerpts are already part of a research dataset and are included solely to explain the method. The main dual-use concern is conceptual: a formal analysis of audit gaming could help auditors design stronger metrics, but it could also help platforms identify weak metrics. We mitigate this risk by releasing only toy code and synthetic data without any platform-specific exploit details or operational advice for evading a real audit. The intended benefit is to improve the design and review of safety measurements used in high-stakes audits, especially in areas where minors may be affected.

\section{Open Science}
\label{sec:open-science}

The artifacts needed to evaluate the paper's core contributions are:
\begin{enumerate}[leftmargin=1.2em]
    \item source code for the deterministic stress test, finite-state verification, SMT query generation, and random-catalog experiments;
    \item generated CSV/JSON outputs used to populate the numerical claims in the paper;
    \item figure source files and rendered figures;
    \item the PRISM model, property file, and raw solver output for the sequential toy audit;
    \item the SMT query generator together with the emitted query instances and solver logs; and
    \item a short note documenting the HateCheck-derived class example and any sensitivity analysis.
\end{enumerate}

The artifact is publicly hosted at:

\noindent\textbf{Artifact repository:} \href{https://github.com/flonat/ccs-2026-formal-verification-artifact}{https://github.com/flonat/ccs-2026-formal-verification-artifact} (a Zenodo DOI snapshot will be issued alongside this preprint).

\noindent The artifact bundle contains:
\begin{itemize}[leftmargin=1.4em]
    \item \path{code/python/} --- \path{catalogs.py} (catalog builders), \path{reproduce.py} (LPs and grid enumeration), \path{verify_smt.py} (Z3 + cvc5 cross-solve), \path{scalability.py} (SMT at $n$ up to $1005$), \path{sensitivity.py} (closed-form $\alpha/\eta$ Pareto), \path{pareto.py} (utility--harm trajectory)
    \item \path{experiments/prism/} --- PRISM-games model, property file, and runner
    \item \path{data/processed/} --- reference CSV/JSON outputs for every numerical claim, including the $15$ \texttt{.smt2} query instances under \path{smt_queries/}, the PRISM trace in \path{prism_results.txt}, the random-catalog summary, and the Pareto-trajectory data
    \item \path{docs/hatecheck_worked_example.md} --- protocol note for the \Cref{sec:hatecheck-mini} instantiation
    \item \path{Dockerfile}, \path{run_all.sh}, \path{pyproject.toml}, \path{uv.lock} --- one-shot reproducibility container (Python~3.13, Z3~4.16.0, cvc5~1.3.3, OpenJDK, PRISM-games~4.8) with a deterministic build
    \item \path{README.md}, \path{EXPECTED_OUTPUTS.md}, \path{BADGES.md} --- evaluator-facing documentation: quick-start (Docker and host), per-headline-number reference values with tolerances, and the mapping to the CCS Artifact Evaluation criteria
\end{itemize}
No credentials are needed; the worked example does not depend on any non-public dataset or service.

\section{Protocol Sensitivity and Harm-Purity Checks}
\label{app:protocol-sensitivity}

We instantiate \Cref{prop:refinement} on the six-variant model with two validation thresholds, $\rho=0.70$ and $\rho=0.90$. The lower threshold keeps both harmful pairs merged, and the higher threshold retains the $H1$ pair but splits $H2$ into singleton harmful classes. The generated appendix table reports the induced harmful partition together with the maximum finite-state certificate violation on the $0.05$ simplex grid.

\input{tables/protocol_sensitivity.tex}

This check confirms the monotonicity claim pointwise: every variant's envelope score under $\rho=0.90$ is at most its score under $\rho=0.70$. In the stricter protocol, the within-class gaming surface for $H2$ disappears altogether; therefore, the fragile and envelope metrics coincide on that class. The trade-off is a much smaller global coverage term $\alpha=0.10$, which makes the coarse global corollary looser, even though the observed grid violation remains $0$.

\subsection{Operational harm-purity verification}
\label{app:operational-harm-purity}

Before applying the ideal-case certificate to a real audit, the auditor should run a per-variant harm-purity check (a class-average agreement rate is not adversarially valid under $\Strategies=\Delta(\variants)$, because the platform can concentrate exposure on the few disputed variants):
\begin{enumerate}[leftmargin=1.3em]
    \item sample every variant in every harmful class that can receive platform exposure (and every audit-negative class whose mass-bearing variants might mask false negatives, since a wrong audit-negative label also contributes to $\Dmass(x)$);
    \item obtain at least $N\geq 3$ independent annotations per sampled variant using the same harm definition that underwrites the audit;
    \item for each variant, declare it harm-aligned when the annotator majority matches the published class label, and compute a per-variant error upper bound $\epsilon_v$ (e.g.\ a one-sided binomial confidence bound from the $N$ annotations);
    \item if \emph{every} sampled variant in the class passes the predeclared harm-purity test, route the class through the ideal certificate and report the resulting classwise envelope coverage values $\hat\alpha_c$;
    \item otherwise retain the class in the protocol but route it through \Cref{thm:disagreement} using the per-variant or exposure-weighted disagreement bound $\bar\eta_c=\max_{v\in c}\epsilon_v$ (or a tighter exposure-capped variant when caps are part of the protocol).
\end{enumerate}
This procedure does not prove latent harm-purity in a metaphysical sense; it operationalizes the assumption so the reader can see when the ideal-case certificate is justified and when the disagreement-slack certificate should replace it.

\paragraph{Status of this procedure.}
The five-step workflow above is a \emph{proposed} auditor-facing procedure; we have not executed it as an empirical annotation study in this paper. The HateCheck example in \Cref{tab:hatecheckclass} is a worked protocol-construction illustration, not an annotation experiment. Validating the procedure end-to-end on real audit data (recruiting annotators, collecting per-variant harm labels, computing inter-annotator agreement, and calibrating the published slack $\bar\eta$) is future work.

\paragraph{Operational cost and feasibility.}
The dominant operational costs are annotator recruitment, training, and the per-variant labeling itself. For a harmful class with $V_c$ candidate variants and $N\geq 3$ annotations per variant, a single class consumes $N\cdot V_c$ harm-label judgments; an audit covering the worst-case classes from a published protocol with $|\classes_h|$ harmful classes consumes on the order of $3 V_c |\classes_h|$ judgments at the lowest defensible $N$. At a sustained throughput of one judgment per minute per annotator and three independent annotators, a class with $V_c=20$ variants requires roughly one annotator-hour to clear the basic harm-purity check; class-purity targets above $0.80$ majority agreement or Fleiss'-$\kappa$ thresholds above $0.6$, raise this floor whenever inter-annotator disagreement requires adjudication. Where annotator time is scarce, the auditor can prioritize the harmful classes that the platform's deployed strategy actually concentrates mass on (the variants that drive $\Dmass(x)$ in \Cref{thm:disagreement}) rather than the full published protocol. A principled empirical-validation plan for the workflow---a small pilot on real audit data, followed by a staged rollout that tracks per-variant error bounds across one full audit cycle---is the natural next step beyond this paper, and the design of that plan is itself a contribution we leave open.

\paragraph{Worked calibration arithmetic.}
To make the machinery of \Cref{thm:disagreement} concrete, we record a worked calibration under plausible values that an auditor would observe after running the five steps. Suppose the audit publishes a harmful class with coverage $\hat\alpha_c=0.90$, and that the per-variant annotation study on that class yields a per-variant Bernoulli error bound $\max_v \epsilon_v \le 0.18$ across every variant in the class. With an admissible within-class support bound $\bar s=1$, the published slack is then $\bar\eta\le 0.18$. \Cref{thm:disagreement} instantiates as
\[
    \TrueHarm(x)\le \tfrac{1}{0.90}\,\Metric_{\Env(m)}(x) + 0.18
\]
for every platform strategy $x$. At an audit budget of $\tau=0.20$, the certified ceiling on true harmful exposure is $(0.20/0.90)+0.18\approx 0.402$. This is looser than the ideal-case bound $0.20/0.90\approx 0.222$; the gap of $0.18$ is exactly the price of moving from assumed harm-purity to a per-variant error bound $\epsilon_c\le 0.18$. A class-average agreement statistic such as a HateCheck-style $\kappa=0.82$ does \emph{not} justify $\bar\eta=0.18$ on its own (see the remark after \Cref{cor:disagreement}): the platform can concentrate exposure on the small fraction of disputed variants and inflate $\Dmass(x)$ well past $1-\kappa$. To use the same numerical figure, the auditor must convert the aggregate $\kappa$ into a per-variant upper-confidence bound (e.g.\ a Wilson or Clopper--Pearson interval on each variant's miscoverage rate) and publish $\max_v \epsilon_v$ rather than the class average. Auditors should report $\hat\alpha_c$, the per-variant error bounds, and the resulting $\bar\eta$ explicitly so that the certificate they publish is reproducible from the same numbers.

\subsection{Sensitivity to coverage and disagreement}
\label{app:sensitivity}

Because $\hat\alpha$ and $\bar\eta$ are estimated and not chosen, an auditor needs to understand how the certificate degrades under realistic measurement error. \Cref{cor:disagreement} gives the worst-case ceiling at audit budget $\tau$ in closed form: $\TrueHarm(x)\le \tau/\hat\alpha + \bar\eta$. \Cref{fig:sensitivity-alpha-eta} traces this ceiling at the budget $\tau=0.20$ used in our deterministic and random-catalog experiments. The certificate is meaningful (i.e.\ strictly below the trivial bound $H^\star\le 1$) whenever $\hat\alpha > \tau/(1-\bar\eta)$; for $\bar\eta=0.20$ this requires $\hat\alpha > 0.25$, and for $\bar\eta=0.10$, $\hat\alpha > 0.222$ suffices. The bound is also slowly varying in $\bar\eta$: at $\hat\alpha=0.85$, doubling the disagreement slack from $0.10$ to $0.20$ moves the ceiling from $0.335$ to $0.435$, a change of $0.10$ that exactly tracks the additive $\bar\eta$ term. The implication for deployment is that coverage estimation has greater leverage---the multiplicative $1/\hat\alpha$ term punishes low-coverage classes harshly---and that any audit reporting $\bar\eta\le 0.20$ alongside $\hat\alpha\ge 0.75$ delivers a certified ceiling below $0.47$ on true harmful exposure. The reproduction script for this figure is \path{code/python/sensitivity.py}.

\begin{figure}[t]
    \centering
    \includegraphics[width=0.6\textwidth]{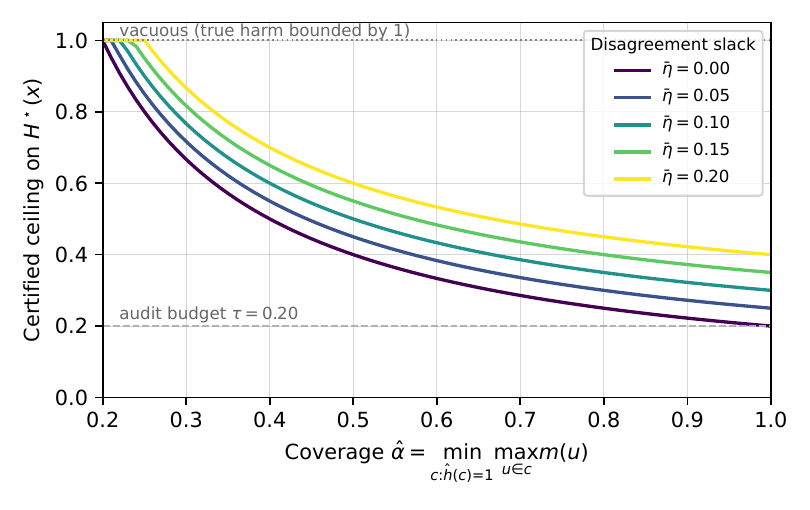}
    \caption{Certified ceiling on true harmful exposure $\TrueHarm(x)\le \tau/\hat\alpha + \bar\eta$ as a function of coverage $\hat\alpha$, at audit budget $\tau=0.20$, for five values of disagreement slack $\bar\eta$. The ceiling is clipped at $1.0$, the trivial bound. The certificate is meaningful in the region below the dotted line.}
    \Description{Five monotonically decreasing curves on a plot of certified true-harm ceiling against coverage. The lowest curve is the disagreement-slack-zero case; higher curves correspond to disagreement slacks of 0.05, 0.10, 0.15, and 0.20. All curves drop from the trivial vacuous ceiling of 1.0 to below 0.5 as coverage rises from 0.20 to 1.00.}
    \label{fig:sensitivity-alpha-eta}
\end{figure}

\section{Generative AI Usage}

During the preparation of this work, the authors used Anthropic's Claude (Sonnet 4.5 and Opus 4.7, via the Claude Code interface) to assist with two substantive tasks:

\paragraph{Writing.} Claude was used to suggest structural reorderings and to check English usage. All formal content---definitions, propositions, proofs, and the threat model---was written and checked by the authors.

\paragraph{Code.} The Python scripts in \path{code/python/} (catalog construction, linear-program reproduction, SMT query generation, cross-solver replay, scalability sweep, and figure generation) were drafted with Claude assistance and subsequently verified by the authors. The PRISM-games model and property file were manually created. All numerical claims in the paper were re-derived from these scripts prior to submission.

\paragraph{Validation.} The authors verified AI-assisted content by (i) re-running every computational stage end-to-end and checking outputs against \path{EXPECTED_OUTPUTS.md} in the artifact; (ii) cross-replaying every SMT query in cvc5 as an independent solver, with $15/15$ verdict agreement; (iii) verifying every cited reference against CrossRef and OpenAlex DOI lookups; and (iv) line-by-line review of every prose passage. No citations, numerical claims, theorem statements, or experimental results in the paper were generated by AI without independent human verification against the underlying source or computation; the authors take full responsibility for the content of the published article.

\end{document}

%% file: tables/hatecheck_mini.tex
\begin{table*}[t]
    \centering
    \caption{HateCheck-derived semi-real instantiation at audit budget $\tau=0.20$. Classes and variant labels follow the HateCheck functional-test taxonomy; classifier scores and utilities are synthetic-but-rule-based (see \Cref{sec:hatecheck-mini}). Columns report the platform's class-level exposure mass and the class-level contribution to true harmful exposure under each metric.}
    \label{tab:hatecheck-mini}
    \footnotesize
    \begin{tabular}{@{}p{0.33\columnwidth} c S[table-format=1.3] S[table-format=1.3] S[table-format=1.3] S[table-format=1.3]@{}}
    \toprule
    {Class} & {Harm} & {Frag.~mass} & {Env.~mass} & {Frag.~harm} & {Env.~harm} \\
    \midrule
    derog\_neg\_emote\_h\_racial & 1 & 0.000 & 0.000 & 0.000 & 0.000 \\
    derog\_neg\_emote\_h\_women & 1 & 0.000 & 0.000 & 0.000 & 0.000 \\
    ident\_neutral\_nh & 0 & 0.528 & 0.828 & 0.000 & 0.000 \\
    threat\_dir\_h & 1 & 0.472 & 0.172 & 0.472 & 0.172 \\
    \midrule
    \textit{Total true harm} & & \multicolumn{2}{c}{Measured $M(x)=\tau$} & 0.472 & 0.172 \\
    \bottomrule
    \end{tabular}
\end{table*}

%% file: tables/protocol_sensitivity.tex
\begin{table*}[t]
    \centering
    \caption{Protocol-sensitivity check under two validation thresholds on the six-variant model. Higher $\rho$ refines the published partition and weakly lowers the envelope pointwise; in both cases the envelope certificate remains exact on the grid.}
    \label{tab:protocol-sensitivity}
    \footnotesize
    \begin{tabular}{@{}S[table-format=1.2] S[table-format=1.2] p{0.33\columnwidth} S[table-format=1.3] S[table-format=1.3]@{}}
    \toprule
    {$\rho$} & {$\alpha$} & {Induced harmful partition} & {Fragile max. violation} & {Envelope max. violation} \\
    \midrule
    0.70 & 0.85 & H1: {o,m}; H2: {o,m} & 0.882 & 0.000 \\
    0.90 & 0.10 & H1: {o,m}; H2: {m} | {o} & 0.000 & 0.000 \\
    \bottomrule
    \end{tabular}
\end{table*}